\documentclass[]{article}
 \usepackage{graphics}
\usepackage{amsmath}
\usepackage{amsfonts}
\usepackage{hyperref}
\usepackage{amsthm}
 \usepackage{amscd}
\usepackage{mathrsfs}

\newtheorem{theorem}{Theorem}[section]

\newtheorem{Proposition}[theorem]{Proposition}
\hypersetup{
	colorlinks = true,
	linkcolor  = red,
	citecolor  = green,
	filecolor  = cyan,
	urlcolor   = magenta,}

\newcommand{\diff}{^\mathsf{\prime}}
\numberwithin{equation}{section}

\title{Direct and inverse scattering problems for a first-order system with energy-dependent potentials}
\author{T. Aktosun and R. Ercan\\
Department of Mathematics\\
University of Texas at Arlington\\
Arlington, TX 76019-0408, USA}

\date{}

\begin{document}

\maketitle

\begin{abstract}
The direct and inverse scattering problems on the full line are analyzed for a first-order system of ordinary linear differential equations associated with the derivative nonlinear Schr\"odinger equation and related equations. The system contains a spectral parameter and two potentials, where the potentials are proportional to the spectral parameter and hence are called energy-dependent potentials. Using the two
potentials as input, the direct problem is solved by determining the scattering coefficients and the bound-state
information consisting of bound-state energies, their multiplicities, and the corresponding
norming constants. By using two different methods, the corresponding inverse problem is solved
by determining the two potentials when the scattering data set is used as input.
The first method involves the transformation of the energy-dependent system into two distinct
energy-independent systems. The second method involves the establishment of the so-called
alternate Marchenko system of linear integral equations and the recovery of the
energy-dependent potentials from the solution to the alternate Marchenko system.
\end{abstract}

\section{Introduction}
\label{sec:section1}

In this paper we are interested in analyzing  the direct and inverse scattering problems for the first-order  system
\begin{equation}\label{1.1}
\frac{d}{dx}\begin{bmatrix}
\alpha\\
\noalign{\medskip}
\beta
\end{bmatrix}=
\begin{bmatrix}
-i\zeta^2 & \zeta q(x)\\
\noalign{\medskip}
\zeta r(x) & i\zeta^2
\end{bmatrix}
\begin{bmatrix}
\alpha\\
\noalign{\medskip}
\beta
\end{bmatrix},\qquad x\in\mathbb{R},
\end{equation}
where $\zeta $ is the spectral parameter, $x$ is the spatial coordinate, $\mathbb{R}$ denotes the real axis $(-\infty,+\infty),$ the scalar-valued potentials $q(x)$ and $r(x)$ belong to the Schwartz class,  and $\begin{bmatrix}
\alpha\\
\beta
\end{bmatrix}$ is the wavefunction depending on $\zeta$ and $x.$
Because of the appearance of the spectral parameter $\zeta$ in front of the potentials  $q(x)$ and $r(x)$ in \eqref{1.1}, the system \eqref{1.1} is referred to as an energy-dependent system

As indicated in Section~\ref{sec:section3}, the system \eqref{1.1} is associated with the
first-order AKNS system
\cite{ablowitz149inverse,ablowitz1974inverse,ablowitz1981solitons}
\begin{equation}\label{1.2}
\begin{bmatrix}
\xi\\
\noalign{\medskip}
\eta
\end{bmatrix} \diff=
\begin{bmatrix}
-i\lambda &  u(x)\\
\noalign{\medskip}
\ v(x) & i\lambda
\end{bmatrix}
\begin{bmatrix}
\xi\\
\noalign{\medskip}
\eta
\end{bmatrix}, \qquad x\in\mathbb{R},
\end{equation}
where the prime denotes the $x$-derivative, the spectral parameter $\lambda$ is related to $\zeta$ as
\begin{equation}\label{1.3}
\lambda=\zeta^{2},
\end{equation}
and the potentials $u(x)$ and $v(x)$ are related to $q(x)$ and $r(x)$ appearing in \eqref{1.1} as
\begin{equation}\label{ux}
u(x):=q(x)\,E^{-2},
\end{equation}
\begin{equation}\label{vx}
v(x):=\left(-\displaystyle\frac{i}{2}\,r\diff(x)+\displaystyle\frac{1}{4}\,
q(x)\,r(x)^2\right)E^2,
\end{equation}
with $E$ being the function of $x$ given by
\begin{equation}\label{E}
E:=\exp\left(\frac{i}{2}\int_{-\infty}^{x}dy\,q(y)\,r(y)\right).
\end{equation}
 For notational simplicity, we suppress the dependence of $E$ on $x.$

 As also indicated in Section~\ref{sec:section3}, the system \eqref{1.1} is associated with another energy-independent AKNS system, namely
 \begin{equation}\label{1.2b}
 \begin{bmatrix}
 \gamma\\
 \noalign{\medskip}
\epsilon
 \end{bmatrix}\diff =
 \begin{bmatrix}
 -i\lambda &  p(x)\\
 \noalign{\medskip}
 s(x) & i\lambda
 \end{bmatrix}
 \begin{bmatrix}
 \gamma\\
 \noalign{\medskip}
\epsilon
 \end{bmatrix}, \qquad x\in\mathbb{R},
 \end{equation}
 where the potentials $p(x)$ and $s(x)$ are related to $q(x)$ and $r(x)$ as
 \begin{equation}\label{p}
 p(x):=\left(\frac{i}{2}\,q\diff(x)+\frac{1}{4}\,q(x)^2\,r(x)\right)E^{-2},
 \end{equation}
 \begin{equation}\label{s}
 s(x):=r(x)\,E^2.
 \end{equation}
For contrast with \eqref{1.1}, we refer to \eqref{1.2} and \eqref{1.2b} as the energy-independent systems.

	The goal in the direct scattering problem for \eqref{1.1} is, given the potentials $q(x)$ and $r(x),$ to determine the properties of the so-called scattering coefficients associated with \eqref{1.1} and also to determine the so-called bound-state information related to \eqref{1.1}. The inverse problem for \eqref{1.1} is to determine the potentials $q(x)$ and $r(x)$ appearing in \eqref{1.1} in terms of the input data set containing the scattering coefficients and
the bound-state information. The direct and inverse scattering problems for the energy-independent systems  \eqref{1.2}
and \eqref{1.2b} are well understood \cite{ablowitz149inverse,ablowitz1974inverse,ablowitz1981solitons,novikov1984theory,shabat1972exact} when the potentials $u(x),$ $v(x),$ $p(x),$ and $s(x)$ belong to the Schwartz class.

Our goal in this paper is to solve the direct and inverse scattering problems for \eqref{1.1} with the help of the solutions to the direct and inverse scattering problems for \eqref{1.2} and \eqref{1.2b}. We accomplish our goal by establishing the appropriate transformations among certain particular solutions known as the Jost solutions,  the scattering coefficients, and the bound-state information for \eqref{1.1}, \eqref{1.2}, and \eqref{1.2b}. The bound states for these three systems are not necessarily simple, and the multiplicities of the bound states and the corresponding norming constants are treated without a restriction to simple bound states.

The research presented is based on the second author's doctoral thesis \cite{Ercan2018}
under the supervision of the
first author.
Our paper is organized as follows. In Section~\ref{sec:section2} we establish our notation and
provide the preliminaries needed for our analysis in later sections. We describe the scattering data sets
for the energy-independent systems \eqref{1.2} and \eqref{1.2b}, and this is done by paying particular attention
to the bound states and to their multiplicities and by expressing the corresponding bound-state information
in terms of certain triplets of constant matrices. We present a brief summary of the
solution to the corresponding inverse problems via the so-called Marchenko method.
In Section~\ref{sec:section3} the energy-dependent system \eqref{1.1} is transformed into the two energy-dependent
systems \eqref{1.2} and \eqref{1.2b} The transformations are presented relating the so-called
Jost solutions for \eqref{1.1} to the Jost solutions for \eqref{1.2} and for \eqref{1.2b}.
In Section~\ref{sec:section4} a brief description is presented for the determination of the scattering data set
for \eqref{1.1} when the potentials $q(x)$ and $r(x)$ appearing in \eqref{1.1} are used as input.
In Section~\ref{sec:section5} the inverse problem for \eqref{1.1} is solved by exploiting the relations
among the scattering data sets for \eqref{1.1}, \eqref{1.2}, and \eqref{1.2b}. In Section~\ref{sec:section6} the same inverse problem
is solved by using a method different from the one used in Section~\ref{sec:section5}. This is done by deriving
a system of uncoupled linear integral equations, which we call the alternate
Marchenko system. The scattering data set for \eqref{1.1} is used as input to the alternate Marchenko system,
and the potentials $q(x)$ and $r(x)$ are recovered from the solution to the alternate
Marchenko system.

When the potentials $q(x)$ and $r(x)$ appearing in
\eqref{1.1} depend on an additional parameter $t,$ which can be
interpreted as the time parameter, the system \eqref{1.1} is
replaced with
\begin{equation}\label{1.1c}
\frac{d}{dx}\begin{bmatrix}
\alpha\\
\noalign{\medskip}
\beta
\end{bmatrix}=
\begin{bmatrix}
-i\zeta^2 & \zeta q(x,t)\\
\noalign{\medskip}
\zeta r(x,t) & i\zeta^2
\end{bmatrix}
\begin{bmatrix}
\alpha\\
\noalign{\medskip}
\beta
\end{bmatrix},\qquad x\in\mathbb{R},\quad t>0.
\end{equation}
Via the inverse scattering transform, the system \eqref{1.1c} is related to the system of integrable evolution equations \cite{ablowitz149inverse,ablowitz1981solitons,tsuchida2010new}
\begin{equation}\label{1.2a}
\begin{cases}
iq_t+q_{xx}-i(qrq)_x=0,\\
\noalign{\medskip}
ir_t-r_{xx}-i(rqr)_x=0,
\end{cases}
\qquad   x\in\mathbb{R},\quad t>0,
\end{equation}
which is known as the derivative NLS (nonlinear Schr\"odinger) system. Note that the subscripts in \eqref{1.2a} denote the corresponding partial derivatives. The special case of \eqref{1.2a} when $r(x)=\pm q(x)^{\ast},$ where the asterisk denotes complex conjugation, corresponds to the derivative NLS equation \cite{ablowitz149inverse,ablowitz1981solitons,kaup1978exact}
\begin{equation}\label{1.2c}
iq_t+q_{xx}\pm i(q|q|^2)_x=0.
\end{equation}
The initial-value problem for \eqref{1.2c} was solved by Kaup and Newell \cite{kaup1978exact}
by using the method of inverse scattering transform.
Tsuchida developed \cite{tsuchida2010new}
a method to solve the initial-value problem for
\eqref{1.2a} by deriving an uncoupled system of two linear integral
equations resembling our alternate Marchenko system.

\section{Preliminaries}
\label{sec:section2}

In this section we establish our notation and provide the preliminaries needed in later sections. This is done by presenting a summary of the solution to the direct and inverse problems associated with the energy-independent systems \eqref{1.2} and \eqref{1.2b}, as these are needed to solve the corresponding  inverse problem for \eqref{1.1}.  We state the results mainly for \eqref{1.2}, as the results for \eqref{1.2b} can be stated in a similar way. A proof is omitted because it can be found in \cite{ablowitz1974inverse}.

\begin{theorem}
\label{thm:theorem2.1}
Assume that the potentials $u(x)$ and $v(x)$ appearing in \eqref{1.2} belong to the Schwartz class. We then have the following:

\begin{enumerate}

	 \item[\text{\rm(a)}] The system \eqref{1.2} has the so-called Jost solutions $\psi,$ $\phi,$ $\bar{ \psi},$ $\bar{ \phi}$ satisfying the  respective asymptotics
	\begin{equation}\label{2.1}
	\psi(\lambda,x)=\begin{bmatrix}
	0\\
\noalign{\medskip}
	\displaystyle e^{i\lambda x}
	\end{bmatrix} +o(1),\qquad  x\to+\infty, \quad \lambda\in\mathbb{R},
	\end{equation}
	\begin{equation}\label{2.2}
	\phi(\lambda,x)=\begin{bmatrix}
	e^{-i\lambda x}\\
\noalign{\medskip}
	0
	\end{bmatrix} +o(1),\qquad   x\to-\infty,\quad \lambda\in\mathbb{R},
	\end{equation}
	\begin{equation}\label{2.3}
	\bar{\psi}(\lambda,x)=\begin{bmatrix}
	e^{-i\lambda x}\\
\noalign{\medskip}
	0
	\end{bmatrix} +o(1),\qquad  x\to+\infty,\quad \lambda\in\mathbb{R},
	\end{equation}
	\begin{equation}\label{2.4}
	\bar{\phi}(\lambda,x)=\begin{bmatrix}
	0\\
\noalign{\medskip}
	e^{i\lambda x}
	\end{bmatrix} +o(1),\qquad  x\to-\infty, \quad \lambda\in\mathbb{R}.
	\end{equation}
	We remark that an overbar does not denote complex conjugation. We further remark that
our definition of the Jost solution $\bar{\phi}$ differs by a minus sign from
that used in \cite{ablowitz1974inverse}.

	\item[\text{\rm(b)}]  The six scattering coefficients $T,$ $R,$ $L,$ $\bar{T},$ $\bar{R},$ $\bar{L}$ associated with \eqref{1.2} are obtained from the large x-asymptotics of the Jost solutions as
	\begin{equation}\label{2.5}
	\psi(\lambda,x)=\begin{bmatrix}
	\displaystyle\frac{L(\lambda)}{T(\lambda)}\,e^{-i\lambda x}\\\noalign{\medskip}
	\displaystyle\frac{1}{T(\lambda)}\,e^{i\lambda x}
	\end{bmatrix} +o(1), \qquad   x\to-\infty,\quad \lambda\in\mathbb{R},
	\end{equation}
	\begin{equation}\label{2.6}
	\phi(\lambda,x)=\begin{bmatrix}
	\displaystyle\frac{1}{T(\lambda)}\,e^{-i\lambda x}\\\noalign{\medskip}
	\displaystyle\frac{R(\lambda)}{T(\lambda)}\,e^{i\lambda x}
	\end{bmatrix} +o(1), \qquad   x\to+\infty,\quad \lambda\in\mathbb{R},
	\end{equation}
	\begin{equation}\label{2.7}
	\bar{\phi}(\lambda,x)=\begin{bmatrix}
	\displaystyle\frac{\bar{R}(\lambda)}   {\bar{T}(\lambda)}\,e^{-i\lambda x}\\\noalign{\medskip}
	\displaystyle\frac{1}{\bar{T}(\lambda)}\,e^{i\lambda x}
	\end{bmatrix} +o(1), \qquad   x\to+\infty,\quad \lambda\in\mathbb{R},
	\end{equation}
	\begin{equation}\label{2.8}
	\bar{\psi}(\lambda,x)=\begin{bmatrix}
	\displaystyle\frac{1}{\bar{T}(\lambda)}\,e^{-i\lambda x}\\
\noalign{\medskip}
	\displaystyle\frac{\bar{L}(\lambda)}{\bar{T}(\lambda)}\,e^{i\lambda x}
	\end{bmatrix} +o(1), \qquad  x\to-\infty,\quad \lambda\in\mathbb{R}.
	\end{equation}
	We remark that $T$ and $\bar{T}$ are known as  the transmission coefficients, $R$ and $\bar{R}$ as the reflection coefficients from the right, and $L$ and $\bar{L}$ as the reflection coefficients from the left.

	\item[\text{\rm(c)}] The scattering coefficients are not all independent, and in fact we have
	\begin{equation}\label{2.9} L(\lambda)=-\displaystyle\frac{\bar{R}(\lambda)\,T(\lambda)}{\bar{T}(\lambda)}, \quad \bar{L}(\lambda)=-\displaystyle\frac{R(\lambda)\,\bar{T}(\lambda)}{T(\lambda)},\qquad \lambda\in\mathbb{R}.
	\end{equation}

	\item[\text{\rm(d)}]
	The transmission coefficient $T(\lambda)$ has a meromorphic extension  from the real $\lambda$-axis to the upper-half complex plane $\mathbb{C^+}.$ Similarly, the transmission coefficients $\bar{T}(\lambda)$ has a meromorphic extension from the real $\lambda$-axis to the lower-half complex plane $\mathbb{C^-}.$ Furthermore, we have
	\begin{equation}\label{T}
	T(\lambda)=1+O\left(\frac{1}{\lambda}\right),\qquad \lambda\to\infty \quad \text{\rm{in}}\quad \lambda\in\mathbb{\overline{C^+}},
	\end{equation}
	\begin{equation}\label{T1}
	\bar{T}(\lambda)=1+O\left(\frac{1}{\lambda}\right),\qquad \lambda\to\infty \quad \text{\rm{in}}\quad \lambda\in\mathbb{\overline{C^-}},
	\end{equation}
	where we have defined $\mathbb{\overline{C^{\pm}}}:=\mathbb{C^{\pm}}\cup\mathbb{R}.$
	\item[\text{\rm(e)}]  The bound states for \eqref{1.2}, i.e. column-vector solutions to \eqref{1.2} that are square integrable in $x\in\mathbb{R},$ occur at those $\lambda$-values at which $T(\lambda)$ has poles in $\mathbb{C^+}$ and at those $\lambda$-values at which $\bar{T}(\lambda)$ has poles in $\mathbb{C^-}.$ The number of bound states is finite, and each bound state is not necessarily simple but has a finite multiplicity. We use $\Big\{{\lambda}_j  \Big\}_{j=1}^N$ to denote the set of poles of $T(\lambda)$ in $\mathbb{C^+},$ and we use
$m_j$ to denote the multiplicity of
$\lambda_j,$ i.e. we assume
that the number of linearly-independent square-integrable column-vector solutions to \eqref{1.2} when $\lambda=\lambda_j$ is equal to $m_j.$ Thus, the nonnegative integer $N$ denotes the number of bound states associated with $T(\lambda)$ without counting multiplicities. In a similar way, we use $\Big\{{\bar{\lambda}}_j  \Big\}_{j=1}^{\bar{N}}$ to denote the set of poles of $\bar{T}(\lambda)$ in $\mathbb{C^-},$ and we assume that each $\bar{\lambda}_j$ has multiplicity $\bar{m}_j,$ i.e. the number of linearly-independent square-integrable
column-vector solutions to \eqref{1.2} when $\lambda=\bar{\lambda}_j$ is equal to $\bar{m}_j.$ Thus, the nonnegative integer $\bar{N}$ denotes the number of bound states associated with $\bar{T}(\lambda)$ without counting multiplicities.
\end{enumerate}
\end{theorem}

 For the system \eqref{1.2}, associated with each bound state at $\lambda=\lambda_j$ we have the bound-state norming constants $c_{jk},$ where a double index is used as a subscript with $j=1,...,N$ and $k=0,1,...,m_j-1.$ Similarly, associated with each bound state at $\lambda=\bar{\lambda}_j$ we have the bound-state norming constants $\bar{c}_{jk},$ where $j=1,...,\bar{N}$ and $k=0,1,...,\bar{m}_j-1.$ We refer the reader to $(4.29)$ and $(4.49)$ of \cite{busse2008generalized} for the definitions of  and elaborations for $c_{jk}$ and $\bar{c}_{jk},$ respectively.

We recall that the direct scattering problem for \eqref{1.2} consists of the determination of the set of scattering coefficients
\begin{equation}\label{2.10}
\Big\{T(\lambda), R(\lambda), L(\lambda), \bar{T}(\lambda), \bar{R}(\lambda), \bar{L}(\lambda)\Big\},
\end{equation}
as well as the bound-state information
\begin{equation}\label{2.11}
\Big\{{\lambda}_j,  \big\{c_{jk}\big\}_{k=0}^{m_j-1}\Big\}_{j=1}^N, \quad \Big\{\bar{\lambda}_j,  \big\{\bar{c}_{jk}\big\}_{k=0}^{\bar{m}_j-1}\Big\}_{j=1}^{\bar{N}},
\end{equation}
when the potentials $u(x)$ and $v(x)$ are given.

The solution to the direct scattering problem for \eqref{1.2} can be obtained as follows:

\begin{enumerate}
\item[\text{\rm(a)}] Using the potentials $u(x)$ and $v(x)$ appearing in \eqref{1.2}, when they belong to the Schwartz class,  we uniquely determine the corresponding Jost solutions $\psi,$ $\phi,$ $\bar{ \psi},$ $\bar{ \phi}$ satisfying the asymptotics \eqref{2.1}-\eqref{2.4}. The existence and uniqueness of the Jost solutions are established by converting \eqref{1.2} and \eqref{2.1}-\eqref{2.4} into the corresponding  integral equations for those Jost solutions and then by solving those integral equations iteratively by representing their solutions as uniformly convergent infinite series.

\item[\text{\rm(b)}] From the respective asymptotics \eqref{2.5}-\eqref{2.8} of the four constructed Jost solutions, we  recover the scattering coefficients appearing in the set described in \eqref{2.10}.

\item[\text{\rm(c)}] Since the transmission coefficient $T(\lambda)$ has a meromorphic extension from $\mathbb{R}$ to $\mathbb{C^+}$ and the transmission coefficient $\bar{T}(\lambda)$ has a meromorphic extension from $\mathbb{R}$ to $\mathbb{C^-}$  , by using the respective poles and multiplicities of those poles we determine $N,$ $\bar{N},$ $\Big\{{\lambda}_j  \Big\}_{j=1}^N$ ,  $\Big\{{\bar{\lambda}}_j  \Big\}_{j=1}^{\bar{N}}$ as well as the multiplicities $m_j$ and $\bar{m}_j.$

\item[\text{\rm(d)}] Having obtained the Jost solutions $\psi,$ $\phi,$ $\bar{ \psi},$ $\bar{ \phi}$ and the transmission coefficients $T(\lambda)$ and $\bar{T}(\lambda),$ we determine the norming constants $c_{jk}$ and $\bar{c}_{jk}$ appearing in \eqref{2.11} as described by  $(4.29)$ and $(4.49)$ of \cite{busse2008generalized}.
\end{enumerate}

Having outlined the solution to the direct scattering problem for \eqref{1.2}, let us turn our attention to the corresponding inverse scattering problem. The inverse problem for \eqref{1.2} consists of the determination of the potentials $u(x)$ and $v(x)$ in terms of the input data consisting of the corresponding sets appearing in \eqref{2.10} and \eqref{2.11}. As a result of \eqref{2.9}, we can omit the left reflection coefficients $L(\lambda)$ and $\bar{L}(\lambda)$ from the set in \eqref{2.10}. A solution to this inverse problem can be given as follows:

\begin{enumerate}
	\item[\text{\rm(a)}] Using the quantities appearing in \eqref{2.10} and \eqref{2.11}, we  form the  $2\times2$ matrix-valued function $F(y)$ given by
	\begin{equation}\label{2.12}
	F(y):=\begin{bmatrix}
	0&\bar{\Omega}(y)\\
\noalign{\medskip}
\Omega(y)&0
	\end{bmatrix},
	\end{equation}
	where \begin{equation}\label{2.13}
	\Omega(y):=\frac{1}{2\pi}\int_{-\infty}^{\infty}d\lambda\, R(\lambda)\,e^{i\lambda y}+C\,e^{-A y}\,B,
	\end{equation}
	\begin{equation}\label{2.14}
	\bar{\Omega}(y):=\frac{1}{2\pi}\int_{-\infty}^{\infty}d\lambda\,
\bar{R}(\lambda)\,e^{-i\lambda y}+\bar{C}\,e^{-\bar{A} y}\,\bar{B}.
	\end{equation}
	Here $R(\lambda)$ and $\bar{R}(\lambda)$ are the reflection coefficients from the right appearing in \eqref{2.10} and
	\begin{equation}\label{2.15a}
	C_j:=\begin{bmatrix}
	c_{j(m_j-1)}&c_{j(m_j-2)}&\cdots&c_{j1}&c_{j0}
	\end{bmatrix},
	\end{equation}
	\begin{equation}\label{2.16a}
	\bar{C}_j:=\begin{bmatrix}	\bar{c}_{j(\bar{m}_j-1)}&\bar{c}_{j(\bar{m}_j-2)}&\cdots&\bar{c}_{j1}&\bar{c}_{j0}
	\end{bmatrix},
	\end{equation}
	where we observe that $C_j$ is a row vector with $m_j$ components and $\bar{C}_j$ is a row vector with $\bar{m}_j$ components. The square matrices $A$ and $\bar{A}$  and the column vectors $B$ and $\bar{B}$ appearing in \eqref{2.13} and \eqref{2.14} are constructed as follow. We first form the square matrices $A_j$ and $\bar{A}_j,$ with respective sizes of $m_j \times m_j$ and $\bar{m}_j\times \bar{m}_j,$ as
	\begin{equation}\label{Aa}
	A_j:=\begin{bmatrix}
	-i\lambda_j&-1&0&\cdots&0&0\\
	0&-i\lambda_j&-1&\cdots&0&0\\
	0&0&-i\lambda_j&\cdots&0&0\\
	\vdots&\vdots&\vdots&\ddots&\vdots&\vdots\\
	0&0&0&\cdots&-i\lambda_j&-1\\
	0&0&0&\dots&0&-i\lambda_j
	\end{bmatrix},
	\end{equation}
	\begin{equation}\label{A1a}
	\bar{A}_j:=\begin{bmatrix}
	-i\bar{\lambda}_j&-1&0&\cdots&0&0\\
	0&-i\bar{\lambda}_j&-1&\cdots&0&0\\
	0&0&-i\bar{\lambda}_j&\cdots&0&0\\
	\vdots&\vdots&\vdots&\ddots&\vdots&\vdots\\
	0&0&0&\cdots&-i\bar{\lambda}_j&-1\\
	0&0&0&\cdots&0&-i\bar{\lambda}_j
	\end{bmatrix}.
	\end{equation}
	We note that $A_j$ and $\bar{A}_j$ are in  Jordan canonical forms. Next, we construct the column vector $B_j$ so that its first $(m_j-1)$ components are all zero and its last component, i.e. the $m_j$th component, is equal to one. Similarly, we construct the column vector $\bar{B}_j$ so that its first $(\bar{m}_j-1)$ components are all zero and its last component, i.e. the $\bar{m}_j$th component, is equal to one. Thus, we have
	\begin{equation}\label{2.21a}
	B_j:=\begin{bmatrix}
	0\\ \vdots \\
	0\\
	1
	\end{bmatrix},\qquad j=1,\dots,N,
	\end{equation}
	\begin{equation}\label{2.22a}
	\bar{B}_j:=\begin{bmatrix}
	0\\ \vdots \\
	0\\
	1
	\end{bmatrix},\qquad j=1,\dots,\bar{N}.
	\end{equation}
	The block matrices $A$ and $\bar{A},$ the block row vectors $C$ and $\bar{C},$ and the block column vectors $B$ and $\bar{B}$ are formed in terms of $A_j,$ $\bar{A}_j,$ $B_j,$ $\bar{B}_j,$ $C_j,$ $\bar{C}_j$ as
	\begin{equation}\label{A4}
	A:=\begin{bmatrix}
	A_1&0&\cdots&0&0\\
	0&A_2&\cdots&0&0\\
	\vdots&\vdots&\ddots&\vdots&\vdots\\
	0&0&\cdots&A_{N-1}&0\\
	0&0&\cdots&0&A_N
	\end{bmatrix},
	\end{equation}
	\begin{equation}\label{A5}
	\bar{A}:=\begin{bmatrix}
	\bar{A}_1&0&\cdots&0&0\\
	0&\bar{A}_2&\cdots&0&0\\
	\vdots&\vdots&\ddots&\vdots&\vdots\\
	0&0&\cdots&\bar{A}_{\bar{N}-1}&0\\
	0&0&\cdots&0&\bar{A}_{\bar{N}}
	\end{bmatrix},
	\end{equation}
	\begin{equation}\label{B}
	B=\begin{bmatrix}
	B_1\\
B_2\\
	\vdots\\
	B_N
	\end{bmatrix},\quad \bar{B}=\begin{bmatrix}
	\bar{B}_1\\
\bar{B}_2\\
	\vdots\\
	B_{\bar{N}}
	\end{bmatrix},
	\end{equation}
	\begin{equation}\label{C4}
	C:=\begin{bmatrix}
	C_1&C_2&\cdots&C_N
	\end{bmatrix},\quad\bar{C}:=\begin{bmatrix}
	\bar{C}_1&\bar{C}_2&\cdots&\bar{C}_{\bar{N}}
	\end{bmatrix}.
	\end{equation}
	We refer the reader to \cite{aktosunSymmetries,aktosun2007exact,busse2008generalized} for the details of the construction of the matrix triplets $(A, B, C)$ and $(\bar{A},\bar{B},\bar{C})$ in terms of the scattering data sets appearing in \eqref{2.10} and \eqref{2.11}.

	\item[\text{\rm(b)}] Having formed the matrix $F(y)$ appearing in \eqref{2.12}, we use it as input to the linear $2\times2$ matrix-valued integral equation, known as the Marchenko system of integral equations,
	\begin{equation}\label{2.15}
	K(x,y)+F(x+y)+\int_{x}^{\infty}dz\,K(x,z)\,F(z+y)=0,\qquad x<y,
	\end{equation}
	and we obtain the solution $K(x,y)$ for $-\infty<x<y<+\infty,$ where  $K(x,y)$ is expressed in terms of its entries as\begin{equation}\label{2.16}
	K(x,y):=\begin{bmatrix}
	\bar{K}_1(x,y)&K_1(x,y)\\
\noalign{\medskip}\bar{K}_2(x,y)&K_2(x,y)
	\end{bmatrix}.
	\end{equation}
It is understood that $K(x,y)=0$ when $x>y,$ and hence
\eqref{2.15} is valid only when $x<y.$

	\item[\text{\rm(c)}] Having obtained $K(x,y),$ we use  $K(x,x),$ which is defined as $K(x,x^+),$ in order to recover the potentials $u(x)$ and $v(x)$ via
	\begin{equation}\label{2.17}
	u(x)=-2\,K_1(x,x),
	\end{equation}
	\begin{equation}\label{2.18}
	v(x)=-2\,\bar{K}_2(x,x),
	\end{equation}
	\begin{equation}\label{2.19}
	\int_{x}^{\infty}dz\,u(z)\,v(z)=2\,\bar{K}_1(x,x),
	\end{equation}
	\begin{equation}\label{2.20}
	\int_{x}^{\infty}dz\,u(z)\,v(z)=2\,K_2(x,x).
	\end{equation}
	In order to prevent any confusion, we can use the superscript $(u,v)$ and the superscript $(p,s)$ for the input data sets appearing in \eqref{2.10} and \eqref{2.11} in order to associate them with \eqref{1.2} and \eqref{1.2b}, respectively. In other words, instead of \eqref{2.10} and \eqref{2.11} we can respectively use
	\begin{equation}\label{2.21}
	\Big\{T^{(u,v)}, R^{(u,v)}, L^{(u,v)}, \bar{T}^{(u,v)}, \bar{R}^{(u,v)}, \bar{L}^{(u,v)}\Big\},
	\end{equation}
	\begin{equation}\label{2.22}
	\Big\{{\lambda}^{(u,v)}_j,  \big\{c^{(u,v)}_{jk}\big\}_{k=0}^{m^{(u,v)}_j-1}\Big\}_{j=1}^{N^{(u,v)}}, \quad \Big\{\bar{\lambda}^{(u,v)}_j,  \big\{\bar{c}^{(u,v)}_{jk}\big\}_{k=0}^{\bar{m}^{(u,v)}_j-1}\Big\}_{j=1}^{\bar{N}^{(u,v)}}.
	\end{equation}
\end{enumerate}

We remark that the potentials $p(x)$ and $s(x)$ appearing in \eqref{1.2b} can be recovered from the input data expressed as in \eqref{2.10} and \eqref{2.11}, i.e.
\begin{equation}\label{2.23}
\Big\{T^{(p,s)}, R^{(p,s)}, L^{(p,s)}, \bar{T}^{(p,s)}, \bar{R}^{(p,s)}, \bar{L}^{(p,s)}\Big\},
\end{equation}
\begin{equation}\label{2.24}
\Big\{{\lambda}^{(p,s)}_j,  \big\{c^{(p,s)}_{jk}\big\}_{k=0}^{m^{(p,s)}_j-1}\Big\}_{j=1}^{N^{(p,s)}}, \qquad \Big\{\bar{\lambda}^{(p,s)}_j,  \big\{\bar{c}^{(p,s)}_{jk}\big\}_{k=0}^{\bar{m}^{(p,s)}_j-1}\Big\}_{j=1}^{\bar{N}^{(p,s)}},
\end{equation}
and this can be achieved by using the analogs of the steps (a)-(c) listed above and by starting with the input data set consisting of \eqref{2.23} and \eqref{2.24}.

\section{Transformations}
\label{sec:section3}

In this section we provide the transformations of the energy-dependent system \eqref{1.1} into the two energy-independent systems given in \eqref{1.2} and \eqref{1.2b}, respectively. Thus, we deal with three first-order systems given in \eqref{1.1}, \eqref{1.2}, and \eqref{1.2b}. For each of these systems there are four Jost solutions $\psi,$ $\phi,$ $\bar{ \psi},$ $\bar{ \phi}$ and there are six scattering coefficients $T,$ $R,$ $L,$ $\bar{T},$ $\bar{R},$ $\bar{L}.$ From \eqref{1.1}, \eqref{1.2}, and \eqref{1.2b}, we see that the corresponding coefficient matrices for these three systems all have the same asymptotic limits as $x\to\pm\infty,$ and hence for each of these three systems it is possible to define the corresponding Jost solutions $\psi,$ $\phi,$ $\bar{ \psi},$ $\bar{ \phi}$ in the same way by using the spatial asymptotics given in \eqref{2.1}-\eqref{2.4}. Similarly, for each of these three systems it is possible to define the scattering coefficients  $T,$ $R,$ $L,$ $\bar{T},$ $\bar{R},$ $\bar{L}$ in the same way by using the spatial asymptotics given in \eqref{2.5}-\eqref{2.8}.

By indicating the appropriate pair of potentials appearing in the off-diagonal entries in the corresponding  coefficient matrix, we are able to uniquely identify the Jost solutions and the scattering coefficients for each of the three first-order systems given in \eqref{1.1}, \eqref{1.2}, and \eqref{1.2b}, respectively. For example, we use $ \psi^{(\zeta q,\zeta r)},$ $\phi^{(\zeta q,\zeta r)},$ $\bar{\psi}^{(\zeta q,\zeta r)},$ $\bar{\phi}^{(\zeta q,\zeta r)}$ to denote the four Jost solutions corresponding to \eqref{1.1}; we use $ \psi^{(u,v)},$ $\phi^{(u,v)},$ $\bar{\psi}^{(u,v)},$ $\bar{\phi}^{(u,v)}$ to denote the four Jost solutions corresponding to \eqref{1.2}, etc. Similarly, we use  $T^{(\zeta q, \zeta r)},$ $R^{(\zeta q, \zeta r)},$ $L^{(\zeta q, \zeta r)},$ $\bar{T}^{(\zeta q, \zeta r)},$ $\bar{R}^{(\zeta q, \zeta r)},$ $\bar{L}^{(\zeta q, \zeta r)}$ to denote the scattering coefficients corresponding to \eqref{1.1}; as in \eqref{2.21}  we use  $T^{(u, v)},$ $R^{(u, v)},$ $L^{(u, v)},$ $\bar{T}^{(u, v)},$ $\bar{R}^{(u, v)},$ $\bar{L}^{(u, v)}$ to denote the scattering coefficients corresponding to the system \eqref{1.2}, etc.

In  the next theorem we provide the relations among the Jost solutions to \eqref{1.1} and \eqref{1.2}, respectively.

\begin{theorem}
\label{thm:theorem3.1}
Assume that the potentials $q(x)$ and $r(x)$ appearing in the first-order system \eqref{1.1}   belong to the Schwartz class. Then, the system \eqref{1.1} can be transformed into the system \eqref{1.2}, where the potential pair $(u,v)$ is related to $(\zeta q, \zeta r)$ as in \eqref{ux} and \eqref{vx}, and it also follows that the potentials $u(x)$ and $v(x)$ belong to the Schwartz class. The four Jost solutions corresponding to the respective systems \eqref{1.1} and \eqref{1.2} are related to each other as
\begin{equation}\label{3.1}
\psi^{(\zeta q, \zeta r)}=e^{i\mu/2}\begin{bmatrix}
\sqrt{\lambda}\,E&0\\\noalign{\medskip}\displaystyle\frac{i}{2}\,r(x)\,E&E^{-1}
\end{bmatrix}\psi^{(u,v)},
\end{equation}
\begin{equation}	\label{3.2}
\phi^{(\zeta q, \zeta r)}=\begin{bmatrix}
E&0\\\noalign{\medskip}\displaystyle\frac{i}{2\sqrt{\lambda}}\,r(x)\,E&\displaystyle\frac{1}{\sqrt{\lambda}}\,E^{-1}
\end{bmatrix}\phi^{(u,v)},
\end{equation}
\begin{equation}\label{3.3}
\bar{\psi}^{(\zeta q, \zeta r)}=e^{-i\mu/2}\begin{bmatrix}
E&0\\\noalign{\medskip}\displaystyle\frac{i}{2\sqrt{\lambda}}\,r(x)\,E&\displaystyle\frac{1}{\sqrt{\lambda}}\,
E^{-1}
\end{bmatrix}\bar{\psi}^{(u,v)},
\end{equation}
\begin{equation}\label{3.4}
\bar{\phi}^{(\zeta q, \zeta r)}=\begin{bmatrix}
\sqrt{\lambda}\,E&0\\\noalign{\medskip}\displaystyle\frac{i}{2}\,r(x)\,E&E^{-1}
\end{bmatrix}\bar{\phi}^{(u,v)},
\end{equation}
where $E$ is the quantity defined in \eqref{E} and the scalar constant $\mu$ is given by
	\begin{equation}\label{3.5}
\mu:=\int_{-\infty}^{\infty}dy\,q(y)\,r(y).
\end{equation}
\end{theorem}

\begin{proof}
	We refer the reader to \cite{Ercan2018} for the motivation and the detailed proof. One can directly verify that the transformations given in \eqref{3.1}-\eqref{3.4} are compatible with \eqref{1.1}-\eqref{E} and \eqref{2.1}-\eqref{2.4}.
\end{proof}

Next we provide the relations among the Jost solutions to \eqref{1.1} and to \eqref{1.2b}, respectively.

\begin{theorem}
\label{thm:theorem3.2}
Assume that the potentials $q(x)$ and $r(x)$ appearing in the first-order system \eqref{1.1}   belong to the Schwartz class. Then, the system \eqref{1.1} can be transformed into the system \eqref{1.2b}, where the potential pair $(p,s)$ is related to $(\zeta q, \zeta r)$ as in \eqref{p} and \eqref{s}, and it also follows that the potentials $p(x)$ and $s(x)$ belong to the Schwartz class. The four Jost solutions corresponding to the respective systems \eqref{1.1} and \eqref{1.2b} are related to each other as
\begin{equation}\label{3.6}
\psi^{(\zeta q, \zeta r)}=e^{i\mu/2}\begin{bmatrix}
\displaystyle\frac{1}{\sqrt{\lambda}}\,E&-\displaystyle\frac{i}{2\sqrt{\lambda}}\,q(x)\,E
\\
\noalign{\medskip}0&E^{-1}
\end{bmatrix}\psi^{(p,s)},
\end{equation}
	\begin{equation}\label{3.7}
\phi^{(\zeta q, \zeta r)}=\begin{bmatrix}
E&-\displaystyle\frac{i}{2}\,q(x)\,E\\
\noalign{\medskip}0&\sqrt{\lambda}\,E^{-1}
\end{bmatrix}\phi^{(p,s)},
\end{equation}
\begin{equation}\label{3.8}
\bar{\psi}^{(\zeta q, \zeta r)}=e^{-i\mu/2}\begin{bmatrix}
E&-\displaystyle\frac{i}{2}\,q(x)\,E\\
\noalign{\medskip}0&\sqrt{\lambda}\,E^{-1}
\end{bmatrix}\bar{\psi}^{(p,s)},
\end{equation}
\begin{equation}\label{3.9}
\bar{\phi}^{(\zeta q, \zeta r)}=\begin{bmatrix}
\displaystyle\frac{1}{\sqrt{\lambda}}\,E&-\displaystyle\frac{i}{2\sqrt{\lambda}}\,q(x)\,E\\
\noalign{\medskip}0&E^{-1}
\end{bmatrix}\bar{\phi}^{(p,s)},
\end{equation}
where $E$ and $\mu$  are the quantities appearing in \eqref{E} and \eqref{3.5}, respectively.
\end{theorem}

\begin{proof}
	The proof is similar to the proof of Theorem~\ref{thm:theorem3.1}. We again refer the reader to \cite{Ercan2018} for the motivation and the details. One can directly verify that the transformations given in \eqref{3.6}-\eqref{3.9} are compatible with
\eqref{1.1}, \eqref{1.3}-\eqref{s}, and \eqref{2.1}-\eqref{2.4}.
\end{proof}

In the next theorem we summarize the transformations among the scattering coefficients for the first-order systems \eqref{1.1}, \eqref{1.2}, and \eqref{1.2b}.

\begin{theorem}
\label{thm:theorem3.3}
Assume that the potentials $q(x)$ and $r(x)$ appearing in the first-order system \eqref{1.1}   belong to the Schwartz class. Then, the scattering coefficients for \eqref{1.1} are related to the scattering coefficients for \eqref{1.2} and for \eqref{1.2b} as
\begin{equation}\label{3.10}
T^{(\zeta q, \zeta r)}=e^{-i\mu/2}\,T^{(u, v)}=e^{-i\mu/2}\,T^{(p,s)},
\end{equation}
\begin{equation}\label{3.11}
\bar{T}^{(\zeta q, \zeta r)}=e^{i\mu/2}\,\bar{T}^{(u,v)}=e^{i\mu/2}\,\bar{T}^{(p,s)},
\end{equation}
\begin{equation}\label{3.12}
R^{(\zeta q, \zeta r)}=\frac{e^{-i\mu}}{\sqrt{\lambda}}\,R^{(u,v)}=e^{-i\mu}\,\sqrt{\lambda}\,R^{(p,s)},
\end{equation}
\begin{equation}\label{3.13}
\bar{R}^{(\zeta q, \zeta r)}=e^{i\mu}\,\sqrt{\lambda}\,\bar{R}^{(u,v)}=\frac{e^{i\mu}}{\sqrt{\lambda}}\,\bar{R}^{(p,s)},
\end{equation}
\begin{equation}\label{3.14}
L^{(\zeta q, \zeta r)}=\sqrt{\lambda}\,L^{(u,v)}=\frac{1}{\sqrt{\lambda}}\,L^{(p,s)},
\end{equation}
\begin{equation}\label{3.15}
\bar{L}^{(\zeta q, \zeta r)}=\frac{1}{\sqrt{\lambda}}\,\bar{L}^{(u,v)}=\sqrt{\lambda}\,\bar{L}^{(p,s)},
\end{equation}
where  $\mu$ is the quantity defined in \eqref{3.5}, the quantity $\sqrt{\lambda}$ is the same as the spectral parameter $\zeta$ appearing in \eqref{1.3}, and the superscripts for the scattering coefficients are as indicated in \eqref{2.21} and \eqref{2.23}. 	
\end{theorem}

\begin{proof}
Using \eqref{3.1}-\eqref{3.4} and \eqref{3.6}-\eqref{3.9} in the respective asymptotics \eqref{2.5}-\eqref{2.8} corresponding to each of \eqref{1.1}, \eqref{1.2}, and \eqref{1.2b}, we obtain \eqref{3.10}-\eqref{3.15}. For the details we refer the reader to \cite{Ercan2018}.
\end{proof}

For a proper treatment of the bound states with multiplicities and the treatment of the corresponding norming constants in the presence of non-simple bound states, we refer the reader to  \cite{aktosunSymmetries,aktosun2007exact,busse2008generalized}. As shown from the second terms appearing on the right-hand sides of \eqref{2.13} and \eqref{2.14}, the bound states with multiplicities and their corresponding norming constants can be handled by using the matrix triplets $(A, B, C)$ and $(\bar{A},\bar{B}, \bar{C})$ given in \eqref{A4}-\eqref{C4}.

 From \eqref{3.10} we see that the poles of $T^{(\zeta q, \zeta r)}$ and of $T^{(u,v)}$ are closely related to each other, and similarly from \eqref{3.11} we see that the poles of
 $\bar{T}^{(\zeta q, \zeta r)}$ and of $\bar{T}^{(u,v)}$ are closely related to each other. From \eqref{1.3} and \eqref{3.10} we see that the bound states for \eqref{1.1} related to $T^{(\zeta q, \zeta r)}$ occur at the $\zeta$-values $\zeta_j$ and $-\zeta_j$ for $j=1, \dots, N,$ where $\pm\zeta_j$ and $\lambda_j$ are related to each other as $\lambda_j=(\pm\zeta_j)^2.$ We remark that \eqref{3.10} also implies that $T^{(\zeta q, \zeta r)}$ and $T^{(u,v)}$ have the same number of poles in $\mathbb{C^+}$ and hence we can use $N$ to denote that common number. In a similar way, from \eqref{1.3} and  \eqref{3.11} we see that the bound states for \eqref{1.1} related to $\bar{T}^{(\zeta q, \zeta r)}$ occur at the $\zeta$-values $\bar{\zeta}_j$ and $-\bar{\zeta}_j$ for $j=1, \dots, \bar{N},$ where $\pm\bar{\zeta_j}$ and $\bar{\lambda_j}$ are related to each other as $\bar{\lambda}_j=(\pm\bar{\zeta}_j)^2.$ We also remark that \eqref{3.11} implies that $\bar{T}^{(\zeta q,\zeta r)}$ and $\bar{T}^{(u,v)}$ have the same number of poles in $\mathbb{C^-}$ and hence we can use $\bar{N}$ to denote that common number.  Since $\zeta_j$ and $-\zeta_j$ correspond to the same $\lambda_j,$ from \eqref{3.10} we see that each of $\zeta_j$ and $-\zeta_j$ has multiplicity $m_j$ and in an analogous manner from \eqref{3.11} we see that each of $\bar{\zeta}_j$ and $-\bar{\zeta}_j$ has multiplicity $\bar{m}_j.$

Due to the flexibility and simplicity of including the bound-state norming constants in the row vectors $C$ and $\bar{C}$ appearing in \eqref{C4}, the treatment of the bound states in the solution to the corresponding inverse scattering problems is greatly simplified.

\section{The direct problem for the energy-dependent system}
\label{sec:section4}

In this section we describe the solution to the direct scattering problem for \eqref{1.1}. This involves the determination of the corresponding scattering coefficients $T^{(\zeta q, \zeta r)},$ $R^{(\zeta q, \zeta r)},$ $L^{(\zeta q, \zeta r)},$ $\bar{T}^{(\zeta q, \zeta r)},$ $\bar{R}^{(\zeta q, \zeta r)},$ $\bar{L}^{(\zeta q, \zeta r)}$  and the relevant bound-state information when the pair of potential $(q,r)$ in the Schwartz class is given. The solution is obtained as described below:

\begin{enumerate}
	\item[\text{\rm(a)}] We solve the system \eqref{1.1} and uniquely determine the four Jost solutions $\psi^{(\zeta q,\zeta r)},$ $\phi^{(\zeta q,\zeta r)},$ $\bar{\psi}^{(\zeta q,\zeta r)},$ $\bar{\phi}^{(\zeta q,\zeta r)}$  satisfying the respective asymptotics \eqref{2.1}-\eqref{2.4}. The existence and uniqueness of the corresponding Jost solutions can be proved in the standard manner when $q(x)$ and $r(x)$ belong to the Schwartz class. For the details we refer the reader to \cite{Ercan2018}. For example, \eqref{1.1} and \eqref{2.1} for $\psi^{(\zeta q,\zeta r)}$ can be combined in order to obtain an integral equation for $\psi^{(\zeta q,\zeta r)}$ when $q(x)$ and $r(x)$ belong to the Schwartz class. One can prove that the resulting integral equation can be solved via iteration by representing its solution as a uniformly convergent infinite series. Alternatively, one can use the already known existence and uniqueness of the Jost solution $\psi^{(u,v)}$ associated with \eqref{1.2}, and by using \eqref{3.1} one can conclude the existence and uniqueness of  $\psi^{(\zeta q,\zeta r)}$ associated with \eqref{1.1}. The existence and uniqueness of the remaining Jost solutions  $\phi^{(\zeta q,\zeta r)},$ $\bar{\psi}^{(\zeta q,\zeta r)},$ and $\bar{\phi}^{(\zeta q,\zeta r)}$ can be established in similar manner.

\item[\text{\rm(b)}] We then determine the scattering coefficients $T^{(\zeta q, \zeta r)},$ $R^{(\zeta q, \zeta r)},$ $L^{(\zeta q, \zeta r)},$ $\bar{T}^{(\zeta q, \zeta r)},$ $\bar{R}^{(\zeta q, \zeta r)},$ $\bar{L}^{(\zeta q, \zeta r)}$
 from the Jost solutions  $\psi^{(\zeta q,\zeta r)},$ $\phi^{(\zeta q,\zeta r)},$ $\bar{\psi}^{(\zeta q,\zeta r)},$ $\bar{\phi}^{(\zeta q,\zeta r)}$ constructed in the previous step, by using the respective  asymptotics  described in \eqref{2.5}-\eqref{2.8}. Alternatively, those scattering coefficients can be obtained by using some Wronskian relations among the constructed Jost solutions. We remark that the Wronskian of any two column-vector solutions to \eqref{1.1} is independent of $x,$ which is a consequence of the fact that the coefficient matrix in \eqref{1.1} has a zero trace. For any two column-vector solutions  $\begin{bmatrix}
\alpha_1\\
\beta_1
\end{bmatrix} $ and $\begin{bmatrix}
\alpha_2\\
\beta_2
\end{bmatrix} $ to \eqref{1.1}, we recall that the Wronskian is defined as
\begin{equation}\label{4.1}
\begin{bmatrix}
\begin{bmatrix}
\alpha_1\\
\noalign{\medskip}
\beta_2
\end{bmatrix};\begin{bmatrix}
\alpha_2\\
\noalign{\medskip}
\beta_2
\end{bmatrix}
\end{bmatrix}:=\begin{vmatrix}
\alpha_1&\alpha_2\\
\noalign{\medskip}
\beta_1&\beta_2
\end{vmatrix},
\end{equation}
where the absolute bars on the right-hand side in \eqref{4.1} denote the determinant of a $2\times2$ matrix. By evaluating certain Wronskians of the Jost solutions to \eqref{1.1} at $x=+\infty$ or $x=-\infty$ and by using \eqref{2.1}-\eqref{2.8}, it can directly be verified that we have
\begin{equation}\label{4.2}
\begin{bmatrix}
\psi^{(\zeta q,\zeta r)};\phi^{(\zeta q,\zeta r)}
\end{bmatrix}=-\displaystyle\frac{1}{T^{(\zeta q,\zeta r)}},
\end{equation}
\begin{equation}\label{4.3}
\begin{bmatrix}
\bar{\psi}^{(\zeta q,\zeta r)};\bar{\phi}^{(\zeta q,\zeta r)}
\end{bmatrix}=\displaystyle\frac{1}{\bar{T}^{(\zeta q,\zeta r)}},
\end{equation}
\begin{equation}\label{4.4}
\begin{bmatrix}
\psi^{(\zeta q,\zeta r)};\bar{\phi}^{(\zeta q,\zeta r)}
\end{bmatrix}=-\displaystyle\frac{\bar{R}^{(\zeta q,\zeta r)}}{\bar{T}^{(\zeta q,\zeta r)}}=\displaystyle\frac{L^{(\zeta q,\zeta r)}}{T^{(\zeta q,\zeta r)}},
\end{equation}
\begin{equation}\label{4.5}
\begin{bmatrix}
\bar{\psi}^{(\zeta q,\zeta r)};\phi^{(\zeta q,\zeta r)}
\end{bmatrix}=\displaystyle\frac{R^{(\zeta q,\zeta r)}}{T^{(\zeta q,\zeta r)}}=-\displaystyle\frac{\bar{L}^{(\zeta q,\zeta r)}}{\bar{T}^{(\zeta q,\zeta r)}},
\end{equation}
Hence using \eqref{4.2}-\eqref{4.5} the scattering coefficients can be expressed in terms of the Wronskians of the constructed Jost solutions as
\begin{equation}\label{4.6}
T^{(\zeta q,\zeta r)}=\displaystyle\frac{1}{\begin{bmatrix}
	\phi^{(\zeta q,\zeta r)};\psi^{(\zeta q,\zeta r)}
	\end{bmatrix}},
\end{equation}
\begin{equation}\label{4.7}
\bar{T}^{(\zeta q,\zeta r)}=\displaystyle\frac{1}{\begin{bmatrix}
	\bar{\psi}^{(\zeta q,\zeta r)};\bar{\phi}^{(\zeta q,\zeta r)}
	\end{bmatrix}},
\end{equation}
\begin{equation}\label{4.8}
R^{(\zeta q,\zeta r)}=\displaystyle\frac{\begin{bmatrix}
	\phi^{(\zeta q,\zeta r)};\bar{\psi}^{(\zeta q,\zeta r)}
	\end{bmatrix}}{\begin{bmatrix}
	\psi^{(\zeta q,\zeta r)};\phi^{(\zeta q,\zeta r)}
	\end{bmatrix}},
\end{equation}
\begin{equation}\label{4.9}
\bar{R}^{(\zeta q,\zeta r)}=\displaystyle\frac{\begin{bmatrix}
	\bar{\phi}^{(\zeta q,\zeta r)};\psi^{(\zeta q,\zeta r)}
	\end{bmatrix}}{\begin{bmatrix}
	\bar{\psi}^{(\zeta q,\zeta r)};\bar{\phi}^{(\zeta q,\zeta r)}
	\end{bmatrix}},
\end{equation}
\begin{equation}\label{4.10}
L^{(\zeta q,\zeta r)}=\displaystyle\frac{\begin{bmatrix}
	\psi^{(\zeta q,\zeta r)};\bar{\phi}^{(\zeta q,\zeta r)}
	\end{bmatrix}}{\begin{bmatrix}
	\phi^{(\zeta q,\zeta r)};\psi^{(\zeta q,\zeta r)}
	\end{bmatrix}},
\end{equation}
\begin{equation}\label{4.11}
\bar{L}^{(\zeta q,\zeta r)}=\displaystyle\frac{\begin{bmatrix}
	\phi^{(\zeta q,\zeta r)};\bar{\psi}^{(\zeta q,\zeta r)}
	\end{bmatrix}}{\begin{bmatrix}
	\bar{\psi}^{(\zeta q,\zeta r)};\bar{\phi}^{(\zeta q,\zeta r)}
	\end{bmatrix}}.
\end{equation}

\item[\text{\rm(c)}] By using \eqref{3.10} and Theorem~\ref{thm:theorem2.1} one can prove that the transmission coefficient $T^{(\zeta q,\zeta r)}$ is meromorphic in $\lambda\in\mathbb{C^+}$ with a finite number of poles there corresponding to the bound states. Similarly, by using \eqref{3.11} and Theorem~\ref{thm:theorem2.1} one can prove that the transmission coefficient $\bar{T}^{(\zeta q,\zeta r)}$ is meromorphic in  $\lambda\in\mathbb{C^-}$ with a finite number of the poles there corresponding to the bound states. The location and multiplicities of such bound-state poles in $\mathbb{C^+}$ and $\mathbb{C^-},$ respectively, are determined as a result of the uniqueness of the meromorphic extensions from $\mathbb{R}$ to  $\mathbb{C^+}$ and $\mathbb{C^-},$ respectively. The norming constants $c_{jk}$ and $\bar{c}_{jk}$ can then be constructed in a similar way as in  $(4.29)$ and $(4.49)$ of \cite{busse2008generalized}.	
\end{enumerate}

\section{The inverse scattering problem for the energy-dependent system}
\label{sec:section5}

In this section we provide an outline for the solution to the inverse scattering problem for \eqref{1.1}. This involves the recovery of the potentials $q(x)$ and $r(x)$ appearing in \eqref{1.1} from the input data  consisting of the set of scattering coefficients
\begin{equation}\label{5.1}
 \Big\{T^{(\zeta q, \zeta r)}, R^{(\zeta q, \zeta r)}, L^{(\zeta q, \zeta r)}, \bar{T}^{(\zeta q, \zeta r)}, \bar{R}^{(\zeta q, \zeta r)}, \bar{L}^{(\zeta q, \zeta r)}\Big\},
\end{equation}
and the bound-state information
\begin{equation}\label{5.2}
\Big\{\pm\zeta_j,  \big\{c^{(\zeta q, \zeta r)}_{jk}\big\}_{k=0}^{m_j^{(\zeta q, \zeta r)}-1}\Big\}_{j=1}^{N^{(\zeta q, \zeta r)}}, \quad \Big\{\pm\bar{\zeta}_j,  \big\{\bar{c}^{(\zeta q, \zeta r)}_{jk}\big\}_{k=0}^{\bar{m}_j^{(\zeta q, \zeta r)}-1}\Big\}_{j=1}^{\bar{N}^{(\zeta q, \zeta r)}}.
\end{equation}

 From \eqref{1.3}, \eqref{3.10}, and \eqref{3.11} we see that $\zeta$ appears as $\zeta^2$ in $T^{(\zeta q, \zeta r)}$ and in $\bar{T}^{(\zeta q, \zeta r)}.$ From \eqref{3.10} it follows that the number of poles for each of $T^{(\zeta q, \zeta r)},$ $T^{(u,v)},$ and $T^{(p,s)}$ in $\mathbb{C^+}$ is the same, i.e. we have
\begin{equation}\label{5.3}
N^{(\zeta q, \zeta r)}=N^{(u,v)}=N^{(p,s)},
\end{equation}
and we can use $N$ to denote their common value. Similarly, from \eqref{3.11} it follows that the number of poles for each of $\bar{T}^{(\zeta q, \zeta r)},$ $\bar{T}^{(u,v)},$ and $\bar{T}^{(p,s)}$ in $\mathbb{C^-}$ is the same, i.e.
\begin{equation}\label{5.4}
\bar{N}^{(\zeta q, \zeta r)}=\bar{N}^{(u,v)}=\bar{N}^{(p,s)},
\end{equation}
and we can use $\bar{N}$ to denote their common value. Note that the value of $e^{-i\mu/2},$ where $\mu$ is the constant appearing in \eqref{3.5}, can uniquely be determined from the large $\zeta$-asymptotics of $T^{(\zeta q, \zeta r)}$ via
\begin{equation}\label{5.5}
T^{(\zeta q, \zeta r)}=e^{-i\mu/2}\left[1+O\left(\frac{1}{\zeta}\right)\right],\qquad \zeta\to\pm\infty ,
\end{equation}
which directly follows from \eqref{1.3}, \eqref{T}, and \eqref{3.10}. Alternatively, $e^{i\mu/2}$ can be determined from  the large $\zeta$-asymptotics of $\bar{T}^{(\zeta q, \zeta r)}$ via
\begin{equation}\label{5.6}
\bar{T}^{(\zeta q, \zeta r)}=e^{i\mu/2}\left[1+O\left(\frac{1}{\zeta}\right)\right],\qquad \zeta\to\pm\infty,
\end{equation}
which directly follows from \eqref{1.3}, \eqref{T1}, and \eqref{3.11}. Let us remark that from \eqref{4.4} and \eqref{4.5} we have
	\begin{equation}\label{5.7}
L^{(\zeta q, \zeta r)}=-\displaystyle \frac{\bar{R}^{(\zeta q, \zeta r)}\,T^{(\zeta q, \zeta r)}}{\bar{T}^{(\zeta q, \zeta r)}}, \quad \bar{L}^{(\zeta q, \zeta r)}=-\displaystyle\frac{R^{(\zeta q, \zeta r)}\,\bar{T}^{(\zeta q, \zeta r)}}{T^{(\zeta q, \zeta r)}},
\end{equation}
and hence, in order to solve the inverse problem, in the input data set instead of \eqref{5.1} we can use its subset
\begin{equation}\label{5.8}
\Big\{T^{(\zeta q, \zeta r)}, R^{(\zeta q, \zeta r)}, \bar{T}^{(\zeta q, \zeta r)}, \bar{R}^{(\zeta q, \zeta r)}\Big\}.
\end{equation}

 We recall that $\pm\zeta_j$ and $\pm\bar{\zeta}_j$ appearing in \eqref{5.2} are related to $\lambda_j$ and $\bar{\lambda}_j$ appearing in \eqref{2.11} as
\begin{equation}\label{5.9}
\left(\pm\zeta_j\right)^2=\lambda_j,\qquad j=1,\dots,N,
\end{equation}
\begin{equation}\label{5.10}
\left(\pm\bar{\zeta}_j\right)^2=\bar{\lambda}_j,\qquad j=1,\dots,\bar{N}.
\end{equation}
 From \eqref{3.10} it also follows that the multiplicity of $\zeta_j^{(\zeta q, \zeta r)}$ and the multiplicity of $\lambda_j$ are equal, i.e. we have
\begin{equation}\label{5.11}
m_j^{(\zeta q, \zeta r)}=m_j^{(u,v)}=m_j^{(p,s)},\qquad j=1,\dots,N.
\end{equation}
We use $m_j$ to denote the common value appearing in \eqref{5.11}. Similarly, from \eqref{3.11} it follows that the multiplicity of $\bar{\zeta}_j^{(\zeta q, \zeta r)}$ and the multiplicity of $\bar{\lambda}_j$ are equal, i.e. we have
\begin{equation}\label{5.12}
\bar{m}_j^{(\zeta q, \zeta r)}=\bar{m}_j^{(u,v)}=\bar{m}_j^{(p,s)},\qquad j=1,\dots,\bar{N}.
\end{equation}
We use $\bar{m}_j$ to denote the common value appearing in \eqref{5.12}.

We can solve the inverse scattering problem for \eqref{1.1} by using the following steps:
\begin{enumerate}

	\item[\text{\rm(a)}] From $T^{(\zeta q, \zeta r)}$ appearing in the data set \eqref{5.1} we can recover the value of $e^{i\mu/2}$ via \eqref{5.5}.

	\item[\text{\rm(b)}] Using the already recovered value of $e^{i\mu/2},$ from \eqref{5.1} we obtain the auxiliary data subset
	\begin{equation}\label{5.13}
	\Big\{R^{(u,v)}, \bar{R}^{(u,v)}, R^{(p,s)}, \bar{R}^{(p,s)}\Big\},
	\end{equation} 	
	where we have
	\begin{equation}\label{5.14}
R^{(u,v)}=\sqrt{\lambda}\,e^{i\mu}\,R^{(\zeta q,\zeta r)},\quad \bar{R}^{(u,v)}=\frac{1}{\sqrt{\lambda}}\,e^{-i\mu}\,\bar{R}^{(\zeta q,\zeta r)},
	\end{equation}
	\begin{equation}\label{5.15}
R^{(p,s)}=\frac{1}{\sqrt{\lambda}}\,e^{i\mu}\,R^{(\zeta q,\zeta r)},\quad \bar{R}^{(p,s)}=\sqrt{\lambda}\,e^{-i\mu}\,\bar{R}^{(\zeta q,\zeta r)},
	\end{equation}
	which are obtained from \eqref{3.12} and \eqref{3.13}.

	\item[\text{\rm(c)}] Using \eqref{5.9}-\eqref{5.12} we obtain from \eqref{5.2} the auxiliary bound-state information
	\begin{equation}\label{5.16}
	\Big\{{\lambda}_j,  \big\{c^{(u,v)}_{jk}, c^{(p,s)}_{jk}\big\}_{k=1}^{m_j}\Big\}_{j=1}^N, \quad \Big\{\bar{\lambda}_j,  \big\{\bar{c}^{(u,v)}_{jk}, \bar{c}^{(p,s)}_{jk} \big\}_{k=1}^{\bar{m}_j}\Big\}_{j=1}^{\bar{N}},
	\end{equation}
where we remark that the norming constants appearing in \eqref{5.16} can be obtained from the norming constants appearing in \eqref{5.2} with the help of $(4.29)$ and $(4.49)$ of \cite{busse2008generalized} by using the relations among the Jost solutions given in \eqref{3.1}-\eqref{3.4}.

	\item[\text{\rm(d)}] Since the individual norming constants appearing in \eqref{5.16} can be combined into the row vectors $C^{(u,v)},$ $\bar{C}^{(u,v)},$ $C^{(p,s)},$ $\bar{C}^{(p,s)}$ as in \eqref{2.15a}, \eqref{2.16a}, \eqref{C4},  we can conclude that the input data set consisting of \eqref{5.1} and \eqref{5.2} yields the auxiliary input data set consisting of
	\begin{equation}\label{5.17}
	\Big\{R^{(u,v)}(\lambda), \bar{R}^{(u,v)}(\lambda), R^{(p,s)}(\lambda), \bar{R}^{(p,s)}(\lambda)\Big\},
	\end{equation}
	\begin{equation}\label{5.18}
	\Big\{\big\{\lambda_j\big\}_{j=1}^N, \big\{\bar{\lambda}_j\big\}_{j=1}^{\bar{N}}, C^{(u,v)}, \bar{C}^{(u,v)},  C^{(p,s)}, \bar{C}^{(p,s)} \Big\}.
	\end{equation}
	In fact, with the help of \eqref{Aa} and \eqref{A1a} the sets $\big\{\lambda_j\big\}_{j=1}^N$ and $\big\{\bar{\lambda}_j\big\}_{j=1}^{\bar{N}}$ yield the matrices $A$ and $\bar{A}$ appearing in \eqref{A4} and \eqref{A5}, respectively. Knowing the multiplicities $m_j$ and $\bar{m}_j$ appearing in \eqref{5.11} and \eqref{5.12}, with the help of \eqref{2.21a} and \eqref{2.22a} we are also able to construct the column vectors $B$ and $\bar{B}$ appearing in \eqref{B}. Thus, the input data set consisting of \eqref{5.1} and \eqref{5.2} yields the auxiliary input data consisting of the two sets
	\begin{equation}\label{5.19}
	\Big\{ R^{(u,v)}, \bar{R}^{(u,v)}, \big(A, B, C^{(u,v)} \big),  \big(\bar{A},  \bar{B}, \bar{C}^{(u,v)} \big) \Big\},
	\end{equation}
	\begin{equation}\label{5.20}
	\Big\{ R^{(p,s)}, \bar{R}^{(p,s)}, \big(A, B, C^{(p,s)} \big),  \big(\bar{A},  \bar{B}, \bar{C}^{(p,s)} \big) \Big\}.
	\end{equation}
We remark that the matrices $A, B, \bar{A}, \bar{B}$ appearing in \eqref{5.19} coincide with the respective matrices $A, B, \bar{A}, \bar{B}$  appearing in \eqref{5.20}.

\item[\text{\rm(d)}] Having constructed the auxiliary input data set given in \eqref{5.19} with the help of \eqref{2.13} and \eqref{2.14}, we form the matrix-valued function $F^{(u,v)}(y)$ as in \eqref{2.12}. We then use  $F^{(u,v)}(y)$ as input into the Marchenko integral equation \eqref{2.15} and obtain the solution $K^{(u,v)}(x,y)$ to \eqref{2.15}. From \eqref{2.16} we then observe that we have the four scalar entries $\bar{K}_1^{(u,v)}(x,y),$ $\bar{K}_2^{(u,v)}(x,y),$ $K_1^{(u,v)}(x,y),$ $K_2^{(u,v)}(x,y).$ As seen from \eqref{2.17}-\eqref{2.20} we have
	\begin{equation}\label{5.21}
u(x)=-2\,K_1^{(u,v)}(x,x),
\end{equation}
\begin{equation}\label{5.22}
v(x)=-2\,\bar{K}_2^{(u,v)}(x,x),
\end{equation}
\begin{equation}\label{5.23}
\int_{x}^{\infty}dz\,u(z)\,v(z)=2\,\bar{K}_1^{(u,v)}(x,x),
\end{equation}
\begin{equation}\label{5.24}
\int_{x}^{\infty}dz\,u(z)\,v(z)=2\,K_2^{(u,v)}(x,x).
\end{equation}

\item[\text{\rm(e)}] We now repeat the previous step by using the auxiliary input data set given in \eqref{5.20}. In other words, with the help of \eqref{2.13} and \eqref{2.14} we form $F^{(p,s)}(y)$ as in \eqref{2.12}, use it as input in the Marchenko integral equation \eqref{2.15}, and obtain the solution $K^{(p,s)}(x,y)$ to \eqref{2.15}.  As seen from \eqref{2.16} we then have the four scalar entries $\bar{K}_1^{(p,s)}(x,y),$ $\bar{K}_2^{(p,s)}(x,y),$ $K_1^{(p,s)}(x,y),$ $K_2^{(p,s)}(x,y).$ Then, as seen from \eqref{2.17}-\eqref{2.20} we get
\begin{equation}\label{5.25}
p(x)=-2\,K_1^{(p,s)}(x,x),
\end{equation}
\begin{equation}\label{5.26}
s(x)=-2\,\bar{K}_2^{(p,s)}(x,x),
\end{equation}
\begin{equation}\label{5.27}
\int_{x}^{\infty}dz\,p(z)\,s(z)=2\,\bar{K}_1^{(p,s)}(x,x),
\end{equation}
\begin{equation}\label{5.28}
\int_{x}^{\infty}dz\,p(z)\,s(z)=2\,K_2^{(p,s)}(x,x).
\end{equation}

\item[\text{\rm(f)}] Next, we construct the function $E$ appearing in \eqref{E}. This is done as follows. Note that from \eqref{ux} and \eqref{vx} we have
\begin{equation}\label{5.29}
u(x)\,v(x)=-\frac{i}{2}\,r\diff(x)\,q(x)+\frac{1}{4}\,q(x)^2\,r(x)^2,
\end{equation}
and from \eqref{p} and \eqref{s} we have
\begin{equation}\label{5.30}
p(x)\,s(x)=\frac{i}{2}\,q\diff(x)\,r(x)+\frac{1}{4}\,q(x)^2\,r(x)^2.
\end{equation}
From \eqref{5.29} and \eqref{5.30} we get
\begin{equation}\label{5.31}
p(x)\,s(x)-u(x)\,v(x)=\frac{i}{2}\left[q(x)\,r(x)\right]\diff,
\end{equation}
which yields
\begin{equation}\label{5.32}
\int_{x}^{\infty}dz\left( p(z)\,s(z)-u(z)\,v(z)\right) =-\frac{i}{2}\,q(x)\,r(x).
\end{equation}
Comparing \eqref{5.32} with \eqref{5.24} and \eqref{5.28} we see that we can write \eqref{5.32} as
\begin{equation}\label{5.33}
2\left( K_2^{(u,v)}(x,x)-K_2^{(p,s)}(x,x)\right)=\frac{i}{2}\,q(x)\,r(x),
\end{equation}
which yields
\begin{equation}\label{5.34}
2 \int_{x}^{\infty}dz\left( K_2^{(u,v)}(z,z)-K_2^{(p,s)}(z,z)\right)=\frac{i}{2}
\int_{x}^{\infty}dz\,q(z)\,r(z).
\end{equation}
Using \eqref{E} and \eqref{3.5} we can write the right-hand side of \eqref{5.34} as
\begin{equation}\label{5.35}
\frac{i}{2}\int_{x}^{\infty}dz\,q(z)\,r(z)=\frac{i}{2}\,\mu-\frac{i}{2}\int_{-\infty}^{x}dz\,q(z)\,r(z).
\end{equation}
From \eqref{5.34} and \eqref{5.35} we obtain
\begin{equation}\label{5.36}
2\int_{x}^{\infty}dz\left(
K_2^{(u,v)}(z,z)-K_2^{(p,s)}(z,z)\right)=\frac{i}{2}\,\mu-\frac{i}{2}\int_{-\infty}^{x}dz\,q(z)\,r(z).
\end{equation}
Comparing \eqref{5.36} with \eqref{E} we see that $E$ can be constructed via
\begin{equation}\label{5.37}
E=e^{i\mu/2}\,\exp\left(2\int_{x}^{\infty}dz\left(
K_2^{(p,s)}(z,z)-K_2^{(u,v)}(z,z)\right)\right),
\end{equation}
where we recall that $e^{i\mu/2}$ is already constructed via \eqref{5.5} or \eqref{5.6}.

\item[\text{\rm(g)}] Having constructed $u(x),$ $s(x),$ and $E,$ we then use \eqref{ux} and \eqref{s} in order to recover $q(x)$ and $r(x)$ via
\begin{equation}\label{5.38}
q(x)=u(x)\,E^2,
\end{equation}
\begin{equation}\label{5.39}
r(x)=s(x)\,E^{-2}.
\end{equation}
Thus, we have completed the construction of $q(x)$ and $r(x)$ from the input data set consisting of \eqref{5.1} and \eqref{5.2}, which completes the solution to the inverse scattering problem for \eqref{1.1}.
\end{enumerate}

\section{The alternate Marchenko method}
\label{sec:section6}

As indicated in Section~\ref{sec:section2}, the solution to the inverse scattering problems for \eqref{1.2} and \eqref{1.2b} can be obtained from the solution to the system of linear  integral equations given in \eqref{2.15}, which is known as the Marchenko system. In Section~\ref{sec:section5} we have seen that the solution to the inverse scattering problem for \eqref{1.1} can be obtained with the help of the solution to the Marchenko systems associated with \eqref{1.2} and \eqref{1.2b}, respectively. In this section we present another method to solve the inverse problem for \eqref{1.1} by formulating a system of linear integral equation directly associated with \eqref{1.1}. We call our method the alternate Marchenko method, and we call the resulting system of linear integral equations the alternate Marchenko system, due to its similarity to the standard Marchenko system of integral equations.

The motivation behind our alternate Marchenko method comes from the analysis by Tsuchida \cite{tsuchida2010new}, who presented a system of linear integral equations resembling our alternate Marchenko system. There are also some differences between our alternate Marchenko system and that of Tsuchida's; for example, Tsuchida's system lacks the symmetry between the $x$ and $y$ variables that exists in the standard Marchenko system. Tsuchida's main interest in developing his method was to solve initial-value problems for
certain integrable evolution equations, in particular the derivative NLS equation and
related equations. We have found Tsuchida's formulation not very intuitive and not easy to comprehend because it is not clear how the scattering theory is used in the derivation and how Tsuchida's gauge transformations are implemented. Nevertheless, we would like to emphasize the importance of Tsuchida's contribution. Since our main concentration is on the analysis of the direct and inverse scattering problems for \eqref{1.1}, our work presented here may help to clarify the idea behind Tsuchida's formulation, by providing a clearer method to derive the alternate Marchenko system. For further details and elaborations on our method, we refer the reader to \cite{Ercan2018}.

In order to derive the alternate Marchenko system, we first present some auxiliary results.

\begin{Proposition}
\label{prop:proposition6.1}
Assume that the potentials $q(x)$ and $r(x)$ appearing in the first-order system \eqref{1.1}   belong to the Schwartz class. Let $u(x),$ $v(x),$ $p(x),$ $s(x)$ be the potentials defined as in \eqref{ux}, \eqref{vx}, \eqref{p}, \eqref{s}, respectively. Then, the zero-energy Jost solutions $\psi^{(u,v)}(0,x)$ and  $\bar{\psi}^{(u,v)}(0,x)$ to \eqref{1.2} and the zero-energy Jost solutions $\psi^{(p,s)}(0,x)$ and  $\bar{\psi}^{(p,s)}(0,x)$  to \eqref{1.2b} are given by
\begin{equation}\label{6.1}
\begin{bmatrix}
\psi_1^{(u,v)}(0,x)\\
\noalign{\medskip}\psi_2^{(u,v)}(0,x)
\end{bmatrix}:=\psi^{(u,v)}(0,x)=\begin{bmatrix}
-e^{-i\mu/2}\,E^{-1} \displaystyle\int_{x}^{\infty}dy\,q(y)\\
\noalign{\medskip}
e^{-i\mu/2}\,E
\left( 1+\displaystyle \frac{i}{2}\,r(x)\displaystyle\int_{x}^{\infty}dy\,q(y)\right)
\end{bmatrix},
\end{equation}
\begin{equation}\label{6.2}
\begin{bmatrix}
\bar{\psi}_1^{(u,v)}(0,x)\\\noalign{\medskip}\bar{\psi}_2^{(u,v)}(0,x)
\end{bmatrix}:=\bar{\psi}^{(u,v)}(0,x)=\begin{bmatrix}
e^{i\mu/2}\,E^{-1}\\
\noalign{\medskip}-\displaystyle\frac{i}{2}\,e^{i\mu/2}\,r(x)\,E
\end{bmatrix},
\end{equation}
\begin{equation}\label{6.3}
\begin{bmatrix}
\psi_1^{(p,s)}(0,x)\\
\noalign{\medskip}\psi_2^{(p,s)}(0,x)
\end{bmatrix}:=\psi^{(p,s)}(0,x)=\begin{bmatrix}
\displaystyle\frac{i}{2}\,e^{-i\mu/2}\,q(x)\,E^{-1}\\
\noalign{\medskip}e^{-i\mu/2}\,E
\end{bmatrix},
\end{equation}
\begin{equation}\label{6.4}
\begin{bmatrix}
\bar{\psi}_1^{(p, s)}(0,x)\\\noalign{\medskip}\bar{\psi}_2^{(p,s)}(0,x)
\end{bmatrix}:=\bar{\psi}^{(p,s)}(0,x)=\begin{bmatrix}
e^{i\mu/2}\,E^{-1}\left(1-\displaystyle\frac{i}{2}\,q(x)\displaystyle\int_{x}^{\infty}dy\,r(y)\right)\\
\noalign{\medskip}-e^{i\mu/2}\,E\displaystyle\int_{x}^{\infty}dy\,r(y)
\end{bmatrix},
\end{equation}
where $E$ and $\mu$ are the quantities defined in \eqref{E} and \eqref{3.5}, respectively.
\end{Proposition}

\begin{proof}
One can directly verify that the quantities given in \eqref{6.1} and
\eqref{6.2} with the respective asymptotics in \eqref{2.1} and \eqref{2.3}
with $\lambda=0$
satisfy \eqref{1.2} with $\lambda=0.$ Similarly, one can directly verify
that the quantities given in \eqref{6.3} and
\eqref{6.4} with the respective asymptotics in \eqref{2.1} and \eqref{2.3}
with $\lambda=0$
satisfy \eqref{1.2b} with $\lambda=0.$
Alternatively, one can directly solve \eqref{1.2} with $\lambda=0$
using the respective asymptotics \eqref{2.1} and \eqref{2.3}
with $\lambda=0$ and directly obtain the zero-energy Jost solutions
given in \eqref{6.1} and \eqref{6.2}. In a similar way,
one can directly construct the zero-energy Jost solutions
to \eqref{1.2b} with $\lambda=0$ and obtain
\eqref{6.3} and \eqref{6.4}.
\end{proof}

We remark that from \eqref{6.1}-\eqref{6.4} we obtain
\begin{equation}\label{6.5}
q(x)=e^{i\mu}\,\frac{d}{dx}\left[\displaystyle\frac{\psi_1^{(u,v)}(0,x)}{\bar{\psi}_1^{(u,v)}(0,x)}\right],
\end{equation}
\begin{equation}\label{6.6}
r(x)=e^{-i\mu}\,\frac{d}{dx}\left[\displaystyle\frac{\bar{ \psi}_2^{(p,s)}(0,x)}{\psi_2^{(p,s)}(0,x)}\right],
\end{equation}
where we recall that the subscripts 1 and 2 in \eqref{6.1}-\eqref{6.6} refer to the first and second components in the respective Jost solutions.

It is already known \cite{ablowitz149inverse,ablowitz1974inverse,ablowitz1981solitons} that the entries appearing in \eqref{2.16} are related to the entries of the corresponding Jost solutions as
 \begin{equation}\label{6.7}
\psi^{u,v)}(\lambda,x)=\begin{bmatrix}
\psi_1^{(u,v)}(\lambda,x)\\\noalign{\medskip}\psi_2^{(u,v)}(\lambda,x)
\end{bmatrix}=\begin{bmatrix}
0\\\noalign{\medskip}e^{i\lambda x}
\end{bmatrix}+\int_{x}^{\infty}dy\begin{bmatrix}
K_1^{(u,v)}(x,y)\\\noalign{\medskip}K_2^{(u,v)}(x,y)
\end{bmatrix}e^{i\lambda y},
\end{equation}
\begin{equation}\label{6.8}
\bar{\psi}^{(u,v)}(\lambda,x)=\begin{bmatrix}
\bar{\psi}_1^{(u,v)}(\lambda,x)\\\noalign{\medskip}\bar{\psi}_2^{(u,v)}(\lambda,x)
\end{bmatrix}=\begin{bmatrix}
e^{-i\lambda x}\\\noalign{\medskip}0
\end{bmatrix}+\int_{x}^{\infty}dy\begin{bmatrix}
\bar{K}_1^{(u,v)}(x,y)\\\noalign{\medskip}\bar{K}_2^{(u,v)}(x,y)
\end{bmatrix}e^{-i\lambda y},
\end{equation}
 \begin{equation}\label{6.9}
\psi^{p,s)}(\lambda,x)=\begin{bmatrix}
\psi_1^{(p,s)}(\lambda,x)\\\noalign{\medskip}\psi_2^{(p,s)}(\lambda,x)
\end{bmatrix}=\begin{bmatrix}
0\\\noalign{\medskip}e^{i\lambda x}
\end{bmatrix}+\int_{x}^{\infty}dy\begin{bmatrix}
K_1^{(p,s)}(x,y)\\\noalign{\medskip}K_2^{(p,s)}(x,y)
\end{bmatrix}e^{i\lambda y},
\end{equation}
\begin{equation}\label{6.10}
\bar{\psi}^{(p,s)}(\lambda,x)=\begin{bmatrix}
\bar{\psi}_1^{(p,s)}(\lambda,x)\\\noalign{\medskip}\bar{\psi}_2^{(p,s)}(\lambda,x)
\end{bmatrix}=\begin{bmatrix}
e^{-i\lambda x}\\\noalign{\medskip}0
\end{bmatrix}+\int_{x}^{\infty}dy\begin{bmatrix}
\bar{K}_1^{(p,s)}(x,y)\\\noalign{\medskip}\bar{K}_2^{(p,s)}(x,y)
\end{bmatrix}e^{-i\lambda y}.
\end{equation}
In the next proposition we relate the potentials $q(x)$ and $r(x)$ to the quantities appearing in the integrands in \eqref{6.7}-\eqref{6.10}.

\begin{Proposition}
\label{prop:proposition6.2}
Assume that the potentials $q(x)$ and $r(x)$ appearing in the first-order system \eqref{1.1}   belong to the Schwartz class. Let $u(x),$ $v(x),$ $p(x),$ $s(x)$ be the potentials defined as in \eqref{ux}, \eqref{vx}, \eqref{p}, \eqref{s}, respectively. Then, the integrands appearing in \eqref{6.7}-\eqref{6.10} are related to the components of the respective zero-energy Jost solutions to \eqref{1.2} and \eqref{1.2b} as
\begin{equation}\label{6.11}
\psi_1^{(u,v)}(0,x)=\int_{x}^{\infty}dy\,K_1^{(u,v)}(x,y),
\end{equation}
\begin{equation}\label{6.12}
\bar{\psi}_1^{(u,v)}(0,x)=1+\int_{x}^{\infty}dy\,\bar{K}_1^{(u,v)}(x,y),
\end{equation}
\begin{equation}\label{6.13}
\psi_2^{(p,s)}(0,x)=1+\int_{x}^{\infty}dy\,K_2^{(p,s)}(x,y),
\end{equation}
\begin{equation}\label{6.14}
\bar{\psi}_2^{(p,s)}(0,x)=\int_{x}^{\infty}dy\,\bar{K}_2^{(p,s)}(x,y).
\end{equation}
\end{Proposition}

\begin{proof}
	By letting $\lambda=0$ in \eqref{6.7}-\eqref{6.10} we directly obtain \eqref{6.11}-\eqref{6.14}.
\end{proof}

Next we present the derivation of our alternate Marchenko system. For the motivation and further details we refer the reader to \cite{Ercan2018}. Inspired by the right-hand sides of \eqref{6.5} and \eqref{6.6} we define the scalar quantities $\mathscr{K}(x,y)$ and $\bar{\mathscr{K}}(x,y)$ as
\begin{equation}\label{6.15}
\mathscr{K}(x,y):=\displaystyle\frac{\displaystyle\int_{y}^{\infty}dz\,K_1^{(u,v)}(x,z)}{1+\displaystyle\int_{x}^{\infty}dz\,\bar{K}_1^{(u,v)}(x,z)},\qquad x<y,
\end{equation}
\begin{equation}\label{6.16}
\bar{\mathscr{K}}(x,y):=\displaystyle\frac{\displaystyle\int_{y}^{\infty}dz\,\bar{K}_2^{(p,s)}(x,z)}{1+\displaystyle\int_{x}^{\infty}dz\,K_2^{(p,s)}(x,z)},\qquad x<y,
\end{equation}
with the understanding that
\begin{equation}\label{6.17}
\mathscr{K}(x,y)=0,\quad \bar{\mathscr{K}}(x,y)=0,\qquad x>y.
\end{equation}
We remark that the $y$-dependence of $\mathscr{K}(x,y)$ and $\mathscr{\bar{K}}(x,y)$ occurs only in the numerators in \eqref{6.15} and \eqref{6.16}. With the understanding that $\mathscr{K}(x,x)$ and $\bar{\mathscr{K}}(x,x)$ denote  $\mathscr{K}(x,x^+)$ and $\bar{\mathscr{K}}(x,x^+)$ respectively, with the help of \eqref{6.11}-\eqref{6.16} we obtain
\begin{equation}\label{6.18}
\mathscr{K}(x,x)=\displaystyle\frac{\displaystyle\int_{x}^{\infty}dz\,K_1^{(u,v)}(x,z)}
{1+\displaystyle\int_{x}^{\infty}dz\,\bar{K}_1^{(u,v)}(x,z)}=\displaystyle
\frac{\psi_1^{(u,v)}(0,x)}{\bar{\psi}_1^{(u,v)}(0,x)},
\end{equation}
\begin{equation}\label{6.19}
\bar{\mathscr{K}}(x,x)=\displaystyle\frac{\displaystyle\int_{x}^{\infty}dz\,
\bar{K}_2^{(p,s)}(x,z)}{1+\displaystyle\int_{x}^{\infty}dz\,K_2^{(p,s)}(x,z)}=
\displaystyle\frac{\bar{\psi}_2^{(p,s)}(0,x)}{\psi_2^{(p,s)}(0,x)}.
\end{equation}
Comparing \eqref{6.5}, \eqref{6.6}, \eqref{6.18}, and \eqref{6.19} we observe that
\begin{equation}\label{6.20}
q(x)=e^{i\mu}\,\frac{d\mathscr{K}(x,x)}{dx},
\end{equation}
\begin{equation}\label{6.21}
r(x)=e^{-i\mu}\,\frac{d\bar{\mathscr{K}}(x,x)}{dx}.
\end{equation}

We would like to show that the scalar quantities $\mathscr{K}(x,y)$ and $\bar{\mathscr{K}}(x,y)$ defined in \eqref{6.15}-\eqref{6.17} satisfy the alternate Marchenko integral system
	\begin{equation}\label{6.22}
\begin{split}
\mathscr{K}(x,y)+&\bar{G}^{(u,v)}(x+y)
\\ &
+\int_{x}^{\infty}dz\int_{x}^{\infty}dt\, \mathscr{K}_t(x,t)\,G^{(u,v)}(t+z)\,\bar{G}_z^{(u,v)}(z+y)=0,
\qquad x<y,
\end{split}
\end{equation}
\begin{equation}\label{6.23}
\begin{split}
\bar{\mathscr{K}}(x,y)+
&G^{(p,s)}(x+y)
\\ &
+\int_{x}^{\infty}dz\int_{x}^{\infty}dt\, \bar{\mathscr{K}}_t(x,t)\,\bar{G}^{(p,s)}(t+z)\,G_z^{(p,s)}(z+y)=0,
\qquad x<y,
\end{split}
\end{equation}
where the subscripts denote the respective partial derivatives and we have defined
\begin{equation}\label{6.24}
G^{(u,v)}(y):=\int_{y}^{\infty}dz\,\Omega^{(u,v)}(z),\quad \bar{G}^{(u,v)}(y):=\int_{y}^{\infty}dz\,\bar{\Omega}^{(u,v)}(z),
\end{equation}
\begin{equation}\label{6.25}
G^{(p,s)}(y):=\int_{y}^{\infty}dz\,\Omega^{(p,s)}(z),\quad \bar{G}^{(u,v)}(y):=\int_{y}^{\infty}dz\,\bar{\Omega}^{(p,s)}(z),
\end{equation}
with the scalar functions $\Omega^{(u,v)}(z),$ $\bar{\Omega}^{(u,v)}(z),$ $\Omega^{(p,s)}(z),$ $\bar{\Omega}^{(p,s)}(z)$ defined as in \eqref{2.13} and \eqref{2.14}
for the potentials pairs $(u,v)$ and $(p,s),$ respectively.
Note that using the input data sets given in \eqref{5.19} and \eqref{5.20} for the potentials pairs $(u,v)$ and $(p,s),$ respectively, from \eqref{2.13} and \eqref{2.14} it follows that
\begin{equation}\label{6.26}
\Omega^{(u,v)}(y):=\frac{1}{2\pi}\int_{-\infty}^{\infty}d\lambda\, R^{(u,v)}(\lambda)\,e^{i\lambda y}+C^{(u,v)}\,e^{-A y}\,B,
\end{equation}
\begin{equation}\label{6.27}
\bar{\Omega}^{(u,v)}(y):=\frac{1}{2\pi}\int_{-\infty}^{\infty}d\lambda\,
\bar{R}^{(u,v)}(\lambda)\,e^{-i\lambda y}+\bar{C}^{(u,v)}\,e^{-\bar{A} y}\,\bar{B},
\end{equation}
\begin{equation}\label{6.28}
\Omega^{(p,s)}(y):=\frac{1}{2\pi}\int_{-\infty}^{\infty}d\lambda\, R^{(p,s)}(\lambda)\,e^{i\lambda y}+C^{(p,s)}\,e^{-A y}\,B,
\end{equation}
\begin{equation}\label{6.29}
\bar{\Omega}^{(p,s)}(y):=\frac{1}{2\pi}\int_{-\infty}^{\infty}d\lambda\,
\bar{R}^{(p,s)}(\lambda)\,e^{-i\lambda y}+\bar{C}^{(p,s)}\,e^{-\bar{A} y}\,\bar{B}.
\end{equation}

In the next theorem we outline the derivation of \eqref{6.22} and \eqref{6.23}.

\begin{theorem}
\label{thm:theorem6.3}
Assume that the potentials $q(x)$ and $r(x)$ appearing in the first-order system \eqref{1.1}   belong to the Schwartz class.  Let $u(x),$ $v(x),$ $p(x),$ $s(x)$ be the potentials defined as in \eqref{ux}, \eqref{vx}, \eqref{p}, \eqref{s}, respectively. Let  $\mathscr{K}(x,y)$ and $\bar{\mathscr{K}}(x,y)$ be scalar quantities defined as in \eqref{6.15}-\eqref{6.17} and let $G^{(u,v)}(y),$ $\bar{G}^{(u,v)}(y),$  $G^{(p,s)}(y),$ $\bar{G}^{(p,s)}(y)$ be the quantities defined as in \eqref{6.24}-\eqref{6.29}. Then, $\mathscr{K}(x,y)$ and $\bar{\mathscr{K}}(x,y)$ satisfy the alternate Marchenko system given in \eqref{6.22} and \eqref{6.23}.
\end{theorem}

\begin{proof}
	We provide a brief outline of the proof here and refer the reader to \cite{Ercan2018} for the motivation and further details. With the
 help of \eqref{2.16} we can write the $2\times2$ Marchenko system in \eqref{2.15} as four coupled scalar equations as
	\begin{equation}\label{6.30}
	\bar{K}_1(x,t)+\int_{x}^{\infty}dz\,K_1(x,z)\,\Omega(z+t)=0,\qquad x<t,
	\end{equation}
	\begin{equation}\label{6.31}
K_1(x,t)+\bar{\Omega}(x+t)+\int_{x}^{\infty}dz\,\bar{K}_1(x,z)\,\bar{\Omega}(z+t)=0,
\qquad x<t,
	\end{equation}
		\begin{equation}\label{6.32}
\bar{K}_2(x,t)+\Omega(x+t)+\int_{x}^{\infty}dz\,K_2(x,z)\,\Omega(z+t)=0,\qquad x<t,
	\end{equation}
		\begin{equation}\label{6.33}
K_2(x,t)+\int_{x}^{\infty}dz\,\bar{K}_2(x,z)\,\bar{\Omega}(z+t)=0,\qquad x<t.
	\end{equation}
In order to get an integral equation involving  $\mathscr{K}(x,y),$ from the numerator of the right-hand side of \eqref{6.15}, we see that we need to write \eqref{6.31} for the potential pair $(u,v)$ and integrate the resulting equation over $t\in(y, +\infty).$ This yields
\begin{equation}\label{6.34}
\begin{aligned}
\int_{y}^{\infty}&dt\,K_1^{(u,v)}(x,t)+\int_{y}^{\infty}dt\,
\bar{\Omega}^{(u,v)}(x+t)\\&+\int_{y}^{\infty}dt\int_{x}^{\infty}dz\,\bar{K}_1^{(u,v)}
(x,z)\,\bar{\Omega}^{(u,v)}(z+t)=0,\qquad x<z<y.
\end{aligned}
\end{equation}
Similarly, in order to get an integral equation involving $\bar{\mathscr{K}}(x,y),$ from the numerator of the right-hand side of \eqref{6.16}, we see that we need to write \eqref{6.32} for the potential pair $(p,s)$ and integrate the resulting equation over $t\in(y, +\infty).$ This yields
\begin{equation}\label{6.35}
\begin{aligned}
\int_{y}^{\infty}&dt\,\bar{K}_2^{(p,s)}(x,t)+\int_{y}^{\infty}dt\,
\Omega^{(p,s)}(x+t)\\
&+\int_{y}^{\infty}dt\int_{x}^{\infty}dz\,K_2^{(p,s)}
(x,z)\,\Omega^{(p,s)}(z+t)=0,\qquad x<z<y.
\end{aligned}
\end{equation}
By taking the derivatives of \eqref{6.24} and \eqref{6.25}, we see that we can express $\Omega^{(u,v)},$ $\bar{\Omega}^{(u,v)},$ $\Omega^{(p,s)},$ $\bar{\Omega}^{(p,s)}$ in terms of $G^{(u,v)},$ $\bar{G}^{(u,v)},$ $G^{(p,s)},$ $\bar{G}^{(p,s)}$ as
\begin{equation}\label{6.36}
\Omega^{(u,v)}(y)=-G_y^{(u,v)}(y), \quad \bar{\Omega}^{(u,v)}(y)=-\bar{G}_y^{(u,v)}(y),
\end{equation}
\begin{equation}\label{6.37}
\Omega^{(p,s)}(y)=-G_y^{(p,s)}(y), \quad \bar{\Omega}^{(p,s)}(y)=-\bar{G}_y^{(p,s)}(y),
\end{equation}
where we use a subscript to indicate the corresponding
derivative. From \eqref{6.15} and \eqref{6.16} we get
\begin{equation}\label{6.38}
\mathscr{K}(x,y)=\left[ \int_{y}^{\infty}dt\,K_1^{(u,v)}(x,t)\right]\left[1+\int_{x}^{\infty}dz\,
\bar{K}_1^{(u,v)}(x,z)\right]^{-1},\qquad x<y,
\end{equation}
\begin{equation}\label{6.39}
\bar{\mathscr{K}}(x,y)=\left[ \int_{y}^{\infty}dt\,\bar{K}_2^{(p,s)}(x,t)\right]\left[1+\int_{x}^{\infty}
dz\,K_2^{(p,s)}(x,z)\right]^{-1},\qquad x<y.
\end{equation}
With the help
 of \eqref{6.36} and \eqref{6.38}, we write \eqref{6.34} as an integral equation involving $\mathscr{K}(x,y),$ $G^{(u,v)}(y),$ and $\bar{G}^{(u,v)}(y).$ Applying integration by parts on the resulting integrals and finally dividing the resulting equation by the denominator of the right-hand side of \eqref{6.15}, we obtain the integral equation \eqref{6.22} in the alternate Marchenko system. The derivation of \eqref{6.23} is obtained in a similar manner. First, by using \eqref{6.37} and \eqref{6.39} we write \eqref{6.35} in terms of $\bar{\mathscr{K}}(x,y),$ $G^{(p,s)}(y),$ and $\bar{G}^{(p,s)}(y).$ Applying integration by parts in the resulting integrals and finally dividing the resulting equation by the denominator of the right-hand side of \eqref{6.16} we obtain the integral equation \eqref{6.23} in the alternate Marchenko system.
\end{proof}

\section*{Acknowledgement}
The second author is grateful to Prof. T. Tsuchida for his help.


\begin{thebibliography}{16}

	%


	\bibitem{ablowitz149inverse}
	M. J. Ablowitz and  P. A. Clarkson, \textit{Solitons, nonlinear evolution equations and inverse scattering}, Cambridge Univ. Press, Cambridge, 1991.
	

\bibitem{ablowitz1974inverse}
	M. J. Ablowitz, D. J. Kaup, A. C. Newell, and H. Segur, \textit{The inverse
		scattering transform-Fourier analysis for nonlinear problems}, Stud. Appl. Math. {\bf{53}}, 249--315 (1974).

\bibitem{ablowitz1981solitons}
	M. J. Ablowitz and H. Segur, \textit{Solitons and the inverse scattering transform}, SIAM, Philadelphia, 1981.


	\bibitem{aktosunSymmetries}
	T. Aktosun, T. Busse, F. Demontis, and C. van der Mee, \textit{Symmetries for exact solutions to the nonlinear Schr\"odinger equation}, J. Phy. A {\bf{43}}, 025202 (2010).


	\bibitem{aktosun2007exact}
	T. Aktosun, F. Demontis, and C. van der Mee, \textit{ Exact solutions to the focusing nonlinear Schr\"odinger equation}, Inverse Problems {\bf{23}}, 2171--2195 (2007).

\bibitem{busse2008generalized}
	T. N. Busse, \textit{Generalized inverse scattering transform for the nonlinear
		Schr\"odinger equation},  Ph.D. thesis, The University of Texas at
	Arlington, 2008.


\bibitem{Ercan2018}
	R. Ercan, \textit{Scattering and inverse scattering on the line for a first-order system with energy-dependent potentials},  Ph.D. thesis, The University of Texas at
	Arlington, 2018.
	


	\bibitem{kaup1978exact}
	D. J. Kaup  and A. C. Newell, \textit{An exact solution for a derivative nonlinear Schr\"odinger equation}, J. Math. Phys. {\bf{19}}, 798--801 (1978).
	
	
	
	
	
	\bibitem{novikov1984theory}
	S. Novikov, S. V. Manakov, L. P. Pitaevskii, and V. E Zakharov, \textit{Theory of
		solitons: the inverse scattering method}, Consultants Bureau, New York, 1984.
	
	\bibitem{tsuchida2010new}
	T. Tsuchida, \textit{New reductions of integrable matrix partial differential equations: $Sp(m)$-invariant systems}, J. Math. Phys. {\bf{51}}, 053511 (2010).
	
	
	\bibitem{shabat1972exact}
	V. E.  Zakharov and A. B.  Shabat,  \textit{Exact theory of two-dimensional self-focusing and one-dimensional self-modulation of waves in nonlinear media}, Soviet Phys. JETP {\bf{34}}, 62--69 (1972).
	
	
	
	
	
	
	
	
	
	
	
	

	
\end{thebibliography}
\end{document}